\documentclass[lettersize,journal]{IEEEtran}
\usepackage{amsmath, amsfonts}
\usepackage{amsthm}
\usepackage{amssymb}
\usepackage{easyReview}
\usepackage{algorithmic}
\usepackage{algorithm}
\usepackage{array}
\usepackage[caption=false,font=normalsize,labelfont=sf,textfont=sf]{subfig}
\usepackage{textcomp}
\usepackage{stfloats}
\usepackage{url}
\usepackage{bm}
\usepackage{verbatim}
\usepackage{graphicx}
\usepackage{cite}
\usepackage{caption}
\usepackage{booktabs}
\makeatletter
\renewcommand{\maketag@@@}[1]{\hbox{\m@th\normalsize\normalfont#1}}%
\makeatother

\allowdisplaybreaks[4]

\newtheorem{Proposition}{Proposition}

\begin{document}

% \title{Reconfigurable Intelligent Surface Empowered ISAC: Joint Optimization of Communication and Radar Performance}

% \title{Sum-Rate Maximization for RIS Empowered ISAC Systems via Coordinated Beamforming}

\title{Coordinated Beamforming for RIS-Empowered ISAC Systems over Secure Low-Altitude Networks}

% \author{Chunjie Wang,~\IEEEmembership{Student Member,~IEEE,} ...
% 	% <-this % stops a space
% \thanks{C. Wang is with the Shenzhen Institute of Advanced Technology, Chinese Academy of Sciences, Shenzhen 518055, China, the University of Chinese Academy of Sciences, Beijing 100049, China (e-mail: cj.wang@siat.ac.cn)}
% }
\author{
        % Xuhui Zhang,~\IEEEmembership{Student Member,~IEEE},
        % Wenchao Liu,~\IEEEmembership{Student Member,~IEEE},\\
        % Jinke Ren,~\IEEEmembership{Member,~IEEE},
        % Huijun Xing,~\IEEEmembership{Student Member,~IEEE},
        % Gui Gui,~\IEEEmembership{Member,~IEEE},\\
        % Yanyan Shen,~\IEEEmembership{Member,~IEEE},
        % and Shuguang Cui,~\IEEEmembership{Fellow,~IEEE}
        Chunjie Wang,
        Xuhui Zhang,
        Wenchao Liu,
        Jinke Ren,
        Huijun Xing,
        Shuqiang Wang,
        and Yanyan Shen

\thanks{Chunjie Wang is with Shenzhen Institute of Advanced Technology, Chinese Academy of Sciences, Shenzhen 518055, China, and also with University of Chinese Academy of Sciences, Beijing 100049, China (e-mail: cj.wang@siat.ac.cn)
(\emph{Corresponding author: Yanyan Shen})
}

\thanks{
Xuhui Zhang and Jinke Ren are with the Shenzhen Future Network of Intelligence Institute, the School of Science and Engineering, and the Guangdong Provincial Key Laboratory of Future Networks of Intelligence, The Chinese University of Hong Kong, Shenzhen, Guangdong 518172, China (e-mail: xu.hui.zhang@foxmail.com; jinkeren@cuhk.edu.cn).
}

\thanks{
Wenchao Liu is with the School of System Design and Intelligent Manufacturing, Southern University of Science and Technology, Shenzhen 518055, China (e-mail: wc.liu@foxmail.com).
}

\thanks{
Huijun Xing is with the Department of Electrical and Electronic Engineering, Imperial College London, London SW7 2AZ, The United Kingdom (e-mail: huijunxing@link.cuhk.edu.cn).
}

\thanks{
Shuqiang Wang and Yanyan Shen are with Shenzhen Institute of Advanced Technology, Chinese Academy of Sciences, Guangdong 518055, China, and also with the Faculty of Computer Science and Control Engineering, Shenzhen University of Advanced Technology, Guangdong 518055, China (e-mail: sq.wang@siat.ac.cn; yy.shen@siat.ac.cn)
}

% \thanks{This work was supported in part by the Shenzhen Science and Technology Program under Grant JCYJ20220818101607015, the Foundation of Key Laboratory of System Control and Information Processing, Ministry of Education, Shanghai, China under Grant Scip20240114, and the National Natural Science Foundation of China under Grant 61503368.
%  (\emph{Corresponding authors: Y. Shen} (e-mail: yy.shen@siat.ac.cn)) }
% \vspace{-18.5pt}
}

\maketitle

\begin{abstract}
Emerging as a cornerstone for next-generation wireless networks, integrated sensing and communication (ISAC) systems demand innovative solutions to balance spectral efficiency and sensing accuracy.
In this paper, we propose a coordinated beamforming framework for a reconfigurable intelligent surface (RIS)-empowered ISAC system, where the active precoding at the dual-functional base station (DFBS) and the passive beamforming at the RIS are jointly optimized to provide communication services for legitimate unmanned aerial vehicles (UAVs) while sensing the unauthorized UAVs. The sum-rate of all legitimate UAVs are maximized, while satisfying the radar sensing signal-to-noise ratio requirements, the transmit power constraints, and the reflection coefficients of the RIS. To address the inherent non-convexity from coupled variables, we propose a low-complexity algorithm integrating fractional programming with alternating optimization, featuring convergence guarantees. Numerical results demonstrate that the proposed algorithm achieves higher data rate compared to disjoint optimization benchmarks.
This underscores RIS's pivotal
% role in harmonizing UAV-centric communication and sensing functionalities for dynamic aerial networks.
% , showcasing the vital 
role in harmonizing communication and target sensing functionalities for low-altitude networks. 
\end{abstract}


\section{Introduction}
\IEEEPARstart {T}{he} advent of beyond-5G (B5G) and 6G networks is catalyzing transformative architectures for low-altitude networks (LANs), which have emerged as vital enablers of the burgeoning low-altitude economy (LAE) applications \cite{10955337, 10879807}.
However, current terrestrial-centric wireless networks struggle to address the channel dynamics and user mobility challenges in the vertical dimension, where the three-dimensional (3D) intelligent beamforming and energy-efficient resource allocation become pivotal solutions to address these challenges \cite{10833672, 10972017}.
Characterized by the deployment of unmanned aerial vehicles (UAVs), vertical take-off and landing vehicles, and autonomous aerial systems operating below 1,000 meters, LANs achieve unprecedented communication guarantees, e.g., low latency for collision avoidance, high reliability for mission-critical transmission, and improved throughput for real-time air traffic control \cite{10693789, 10976653}. 
Despite these advancements, the inherent mobility of UAVs introduces dynamic channel conditions, where line-of-sight (LoS) links are prone to intermittent blockages from urban structures. This uncertainty in air-to-ground connectivity poses significant challenges for maintaining reliable communication quality, particularly in mission-critical operations requiring real-time coordination.

Reconfigurable intelligent surfaces (RISs), also known as intelligent reflecting surfaces (IRSs), as emerging programmable electromagnetic metamaterials, address the persistent challenge of dynamic signal blockages in the UAV-enabled LANs by dynamically manipulating the phase and amplitude of incident electromagnetic waves with sub-wavelength precision, enabling real-time reconstruction of non-LoS (NLoS) links or enhancement of existing LoS paths \cite{9326394, 10288203}.
Unlike conventional solutions requiring active radio frequency components or additional base stations (BSs), the RISs achieve improvements in signal-to-noise ratio (SNR)  while consuming less energy \cite{9140329}, making them uniquely suitable for energy-constrained UAV-enabled systems.
Their adaptive beam-steering capability allows seamless alignment with mobile UAVs even in dense urban environments, reducing outage probability for low-altitude regions through strategic deployment on buildings, towers, and aerial platforms \cite{8959174}.
Furthermore, the RISs provide spatial multiplexing gains by generating orthogonal virtual channels via phase shifting, while synergizing with the UAVs to extend 3D beam coverage \cite{10945642}.
These combined advantages position the RISs as pivotal enablers for ultra-reliable, energy-efficient connectivity in the next-generation LANs.

Meanwhile, the popularity of emerging applications such as ultra-high-resolution video transmission, large-scale Internet of Things (IoT) interconnection, and intelligent driving, is driving the exponential growth of global mobile demands for data \cite{9252924, 10721288}. 
However, the contradiction between the natural limitations of wireless spectrum and the fixed allocation method has become increasingly prominent. The low-band resources are almost saturated. Although the high-band technology can expand the spectrum boundary, it is subject to high hardware costs and signal propagation losses, making it difficult to achieve global coverage. At the same time, the problem of spectrum fragmentation and resource redundancy caused by the independent operation of communication and sensing in traditional wireless networks has further exacerbated the imbalance between supply and demand, making the efficient use of spectrum a core challenge in the evolution of wireless networks \cite{10539120}.
The integrated sensing and communication (ISAC) technology provides a new paradigm for breaking through the bottleneck of network efficiency through deep multiplexing of spectrum resources and coordinated design of signal waveforms \cite{9737357}.
% Its core advantage is that the communication signal and the sensing signal are in the time domain, frequency domain, spatial multi-dimensional fusion, so that a single spectrum resources at the same time carry data transmission and environmental sensing functions, significantly improving spectrum utilization.
The core advantage of the ISAC lies in its ability to utilize the same spectrum resource for both communication and sensing signals, enabling simultaneous data transmission and environmental perception, thereby significantly improving spectrum utilization.

Although recent advances in RIS-empowered ISAC systems demonstrate substantial improvements in throughput and spectrum efficiency,
existing architectures struggle to concurrently maintain robust connectivity for legitimate UAVs (L-UAVs) and precise detection of unauthorized intruders 
% under dynamic channel blockages.
over LANs.
To bridge this gap, we introduce a novel framework where a dual-functional BS (DFBS) detects one point-like unauthorized UAV (U-UAV) and simultaneously provides communication services for multiple L-UAVs via the direct links as well as with the assistance of a RIS.
By jointly optimizing the active beamforming at the DFBS and the passive beamforming at the RIS, we aim to promote the sum-rate of the L-UAVs under the requirement of sensing SNR for detecting the U-UAV, the total transmit power, and the physical restriction of the RIS phase. 
However, due to the non-linear and non-convexity of the formulated problem, traditional optimization methods are difficult to solve the problem efficiently. Therefore, a low-complexity algorithm based on fractional programming (FP) and alternating optimization (AO) methods is proposed. The simulation results demonstrate the advantages of deploying RIS in ISAC systems for LANs and the efficiency of the proposed algorithm.

The following is a summary of key contributions of this work.
\begin{itemize}
    \item We propose a novel coordinated beamforming framework that integrates active beamforming at the DFBS and passive beamforming at the RIS for LANs. Meanwhile, we formulate an optimization problem to maximize the sum-rate of L-UAVs communication while ensuring sensing SNR constraints for detecting the U-UAV.
    \item To address the non-convexity of the optimization problem, we develop an efficient algorithm leveraging FP and AO methods. This algorithm decomposes the problem into three tractable subproblems, enabling iterative updates of the auxiliary variables, the active beamforming and the passive beamforming.
    \item Numerical results demonstrate that the proposed algorithm achieves higher sum-rates of all L-UAVs compared to benchmarks. This underscores the pivotal role of RIS in enhancing spectral efficiency in ISAC systems for LANs.
\end{itemize}

\textit{Organizations:}
The remainder of this paper is organized as follows.
Section II reviews the related works.
In Section III, we introduce the RIS-empowered ISAC system for LANs and formulate the sum-rate maximization problem.
In Section IV, we propose a low-complexity algorithm integrating FP with AO methods.
Section V evaluates the performance of the proposed algorithm.
Finally, Section VI concludes the paper.

\textit{Notations:}
% In this paper, the following notations are used: $ {{\mathbb{C}}^{M\times N}} $ represents the $ M \times N $ complex matrix $ \mathbb{C} $. $ {\cal C}{\cal N}(\mu ,{\sigma ^2}) $ denotes the circularly symmetric complex Gaussian distribution with a mean of $ \mu $ and variance $ {\sigma ^2} $. The symbol $ \mathrm{j} $ is the imaginary unit, satisfying $ {\mathrm{j}^2} = -1 $. For a generic matrix $ {\bf{G}} $, $ {{\bf{G}}^{\mathsf{H}}} $ and $ {{\bf{G}}^{\mathsf{T}}} $ represent the conjugate-transpose and the transpose, respectively. For a vector $ {\bf{w}} $, $ ||{\bf{w}}|| $ signifies the Euclidean norm. The notation $ \mathrm{diag}({\bf{w}}) $ represents a diagonal matrix whose diagonal elements are the elements of the vector $ {\bf{w}} $. $ \Re \{ \cdot \} $ denotes the real part of a complex number.
Throughout this paper, we employ the following mathematical notations.
Let $ \mathrm{j} $ denote the imaginary unit satisfying $ {\mathrm{j}^2} = -1 $.
For a complex number $z$, $ \Re \{ z \} $ represents its real component.
$ {{\mathbb{C}}^{M\times N}} $ represents the $ M \times N $ complex matrix $ \mathbb{C} $.
Given a matrix $ {\bf{G}} $, its conjugate transpose and transpose are denoted by $ {{\bf{G}}^{\mathsf{H}}} $ and $ {{\bf{G}}^{\mathsf{T}}} $, respectively.
For a vector $ {\bf{w}} $, $ ||{\bf{w}}|| $ indicates the Euclidean norm, and $ \mathrm{diag}({\bf{w}}) $ generates a diagonal matrix with entries from $ {\bf{w}} $.
The notation $ {\cal C}{\cal N}(\mu ,{\sigma ^2}) $ signifies a circularly symmetric complex Gaussian distribution with mean $ \mu $ and variance $ {\sigma ^2} $.

\section{Related Works}
    % As wireless communication technology moves from 5G to 6G, 
\subsection{UAV over LANs}
UAVs have emerged as pivotal enablers of LAE applications, providing agile and reconfigurable platforms for emergency communication restoration, real-time media streaming, large-scale data harvesting, and time-sensitive logistics in complex environments \cite{8918497, 9604506}.
Recent studies have explored UAV-enabled LANs to address coverage gaps in disaster response, optimize energy-efficient content delivery for aerial surveillance, and enhance throughput for cargo delivery coordination \cite{9456851}.
In \cite{10458879}, the authors investigated the energy-efficient UAV scheduling and task offloading for the IoT devices under demand uncertainty.
Furthermore, a UAV-enabled heterogeneous mobile edge computing offloading framework was investigated in \cite{10606316}, and the processed data volume was optimized for IoT devices.
In \cite{10632079, 10980172}, the UAVs were dispatched to collect data transmitted by ground users, leveraging age of information-aware task prioritization to minimize data staleness in time-critical scenarios.
The content access policy was optimized in a UAV-enabled content service network to tackle the difficulty of explosive growth of mobile data traffic \cite{10964793}.
Moreover, the training latency was minimized in a UAV-enabled federated learning system for enabling edge intelligence in LANs \cite{10972043}.
However, these prior works primarily focus on the UAV mobility and communication services, overlooking the critical demand for real-time signal sensing to detect and mitigate anomalies in the LANs, such as U-UAV intrusions, environmental interference, and spectrum scarcity under dynamic operational conditions.

\subsection{ISAC}
Inspired by the capability of the ISAC technology to achieve substantial spectral efficiency gains while addressing the sensing requirements of next-generation wireless networks,
\cite{10571789} studied a full-duplex multi-user ISAC system, and optimized the communication and sensing resource allocation.
Meanwhile, in \cite{10556732}, the transmit beamforming for a multiple-input and multiple-output ISAC system with multiple radar targets and communication users was investigated.
% In addition, unmanned aerial vehicles (UAVs) equipped with ISAC technology have been employed in the field of communication.
Benefiting from the ISAC capabilities, recent studies have deployed ISAC in LANs to achieve joint communication-sensing enhancement, such as simultaneous aerial target detection and high-throughput data transmission \cite{10693833, 10602493, chen2025full}.
Firstly, the low-altitude platform was enhanced by the ISAC in \cite{10693833}, where the joint beamforming of communication and sensing was optimized with the assistance of a movable antenna array.
In \cite{10602493}, a resource allocation problem for a multi-UAV assisted ISAC system was studied, where the dual-functional UAVs performed simultaneous sensing of a target and data communication with ground users.
The ISAC over LANs was also capable of providing computing services for edge intelligence computing and sensing requirements \cite{chen2025full}.
The above applications fully validate the effectiveness of ISAC technology in improving spectral efficiency and its multi-scenario adaptability, demonstrating the broad deployment potential from fixed BSs to dynamic UAV platforms.
However, the high sharing of resources also brings complex challenges, such as the spatiotemporal competition between communication data flow and sensing beams may lead to signal mutual interference. Therefore, how to balance spectral efficiency and sensing reliability is still a key problem to be solved urgently.
     
\subsection{RIS-Empowered Systems}	 
Despite the promising advances in low-altitude ISAC systems, their limited spectral efficiency remains a critical bottleneck, particularly in scenarios requiring simultaneous high-precision sensing and ultra-reliable data transmission.  
This challenge motivates the integration of RISs to dynamically reshape wireless channels by tuning phase, amplitude, and polarization states.
To effectively utilize the RIS for improving the ISAC systems, \cite{10613843} proposed two novel designs targeting different scenarios, i.e., spectral-efficient ISAC and energy-efficient ISAC, and the simulation results showed that the two designs can significantly improve the weighted sum-rate performance and communication energy efficiency.
In \cite{10298597}, the authors proposed a secure RIS-empowered ISAC system that utilized a RIS to enhance legitimate communication while treating the radar target as a potential eavesdropper, aiming to maximize the radar output SNR.
In \cite{10453349}, the authors proposed a RIS-assisted UAV-enabled ISAC system, where a dual-functional UAV simultaneously transmitted signals to multiple users and performed sensing missions by jointly optimizing the RIS phase shifts, UAV trajectory, dual-function radar-communication beamforming, and user scheduling to maximize the sum-rate and sensing SNR.
Although existing research has initially explored the application potential of RIS in ISAC systems, there are still challenges in how to efficiently and jointly optimize active and passive beamforming to improve the communication and sensing performance for LANs.  
	
\section{System Model and Problem Formulation}
\begin{figure}[t]
	\centering
	\includegraphics[width=0.9\linewidth]{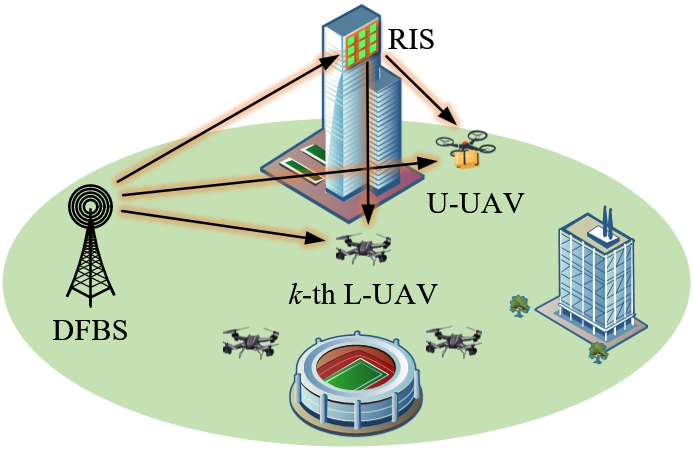}
	\captionsetup{justification=centering}
	\caption{The system model of the RIS-empowered ISAC for the LANs.}
	\label{fig:1}
\end{figure}

%The proposed RIS-assisted ISAC system architecture, as depicted in Fig. \ref{fig:1}, comprises a co-located multi-antenna DFBS that jointly facilitates communication services for $K$ single-antenna communication UAVs (L-UAVs) and performs sensing operations on one target UAV (T-UAV), aided by an $N$-element RIS. The DFBS employs $M$ transmit/receive antennas arranged in uniform linear arrays with half-wavelength spacing. For notational clarity, we define the index sets \({\cal N} = \{1, 2,\ldots, N\}\) and \({\cal K} = \{1, 2,\ldots, K\}\). to represent RIS elements and L-UAVs, respectively.
%We hold the assumption that all UAVs have a location with quasi-varying time\cite{6689509, 9795892, 10683327}. As depicted in Fig. \ref{fig:2}, the total time frame is divided into \(L\) time slots, and define \({\cal L} = \{1, 2,\ldots, L\}\). We assume that the position of the UAVs in each time slot remains unchanged, while at different time slots, the UAVs has a small range of mobility.

As illustrated in Fig. \ref{fig:1}, the considered RIS-empowered ISAC system consists of a multi-antenna DFBS equipped with $ M $ transmit/receive antennas arranged in uniform linear arrays with half-wavelength spacing, an $ N $-element RIS, $ K $ single-antenna L-UAVs, and one single-antenna U-UAV. The DFBS jointly provides communication services for the L-UAVs and senses to detect the U-UAV with the assistance of the RIS. For notational clarity, the index sets \({\cal N} = \{1, 2,\ldots, N\}\) and \({\cal K} = \{1, 2,\ldots, K\}\) are defined for the RIS elements and the L-UAVs, respectively.

We assume that the position of all UAVs is quasi-varying across time duration \cite{9795892, 10683327}.
% As depicted in Fig. \ref{fig:2}, t
The total time duration is divided into \(L\) time slots with the index set \({\cal L} = \{1, 2,\ldots, L\}\), and $ H $ denotes the constant height of the UAVs. Within each time slot \(l\), the position of the $k$-th L-UAV and the U-UAV are given by \({\bf{q}}_k[l] = (q_k^x{[l]}, q_k^y{[l]}, H)\), and \({\bf{q}}_{\mathrm{U}}[l] = (q_{\mathrm{U}}^x{[l]}, q_{\mathrm{U}}^y{[l]}, H)\), respectively, where the horizontal coordinates \((q_k^x{[l]}, q_k^y{[l]})\) and  \((q_{\mathrm{U}}^x{[l]}, q_{\mathrm{U}}^y{[l]})\) are fixed.
% This means that during the duration of one time slot, the UAV maintains a fixed horizontal location, enabling stable communication links and sensing operations.
Across different time slots, the UAVs fly to the next position, e.g., the position of the $ k $-th L-UAV in time slot $ l $ is expressed as \({\bf{q}}_k[l+1] = (q_k^x{[l+1]}, q_k^y{[l+1]}, H)\).\footnote{We assume that the L-UAVs conduct LAE tasks in the flying region, with their positions known across all time slots. Meanwhile, the position of the U-UAV can be estimated via surveillance systems. The goal of the DFBS is to maximize U-UAV sensing precision for ensuring the L-UAV mission safety and securing their data transmission services.}
% Across different time slots, the UAVs undergoes small-range mobility\footnote{In typical ISAC tasks (such as area monitoring, emergency communication relay, and fixed-point data acquisition), UAVs usually need to remain relatively stationary or move at a low speed within a specific area to maintain stable target tracking or communication links. For example, monitoring tasks require UAVs to hover or make small adjustments in position to continuously observe targets; when acting as an aerial base station, UAVs need to avoid large-scale movements that cause severe fluctuations in channels. In addition, the small-scale movement model assumed in this paper is applicable to a short-term optimization window, and its time scale is much smaller than the cycle required for UAVS to complete large-scale maneuvers. During this period, the displacement of UAVs can be approximated as static or linearly changing, thus simplifying system performance analysis and focusing on the joint optimization mechanism of communication and sensing by RIS.}, with its position evolving to \((x{[l+1]}, y{[l+1]}, H)\) in the next slot.
Besides, the positions of the DFBS and the RIS are fixed over the whole time duration, and are given by \({\bf{q}}_{\mathrm{BS}} = (q_{\mathrm{BS}}^x{}, q_{\mathrm{BS}}^y{}, q_{\mathrm{BS}}^z)\) and \({\bf{q}}_{\mathrm{RIS}} = (q_{\mathrm{RIS}}^x{}, q_{\mathrm{RIS}}^y{}, q_{\mathrm{RIS}}^z{})\), respectively.
% \begin{figure}[t]
% 	\centering
% 	\includegraphics[width=0.9\linewidth, height=4.5cm]{"Location of UAVs"}
% 	\captionsetup{justification=centering}
% 	\caption{Illustration of the locations of the UAV across $ L $ time slots.}
% 	\label{fig:2}
% \end{figure}

The transmitted signal of the DFBS in the time slot $ l $ is
\begin{align}
	{\bf{x}}[l] = \sum\limits_{k = 1}^K {{{\bf{w}}_k}[l]{s_{k}}[l]} + {{\bf{w}}_\vartheta}[l]{s_\vartheta}[l],
\end{align}
where $ {s_{k}}[l] \in \mathbb{C} $ denotes the communication signal for the $ k $-th L-UAV, which satisfies $ \mathbb{E}\{ {{{\left| {{s_{k}[l]}} \right|}^2}} \} = 1 $. The corresponding beamforming vector is $ {{\bf{w}}_k}[l] \in {\mathbb{C}^{M \times 1}} $. $ {s_{\vartheta}[l]} \in \mathbb{C} $ is the sensing signal and meets $ \mathbb{E}\{ {{{\left| {{s_{\vartheta}[l]}} \right|}^2}} \} = 1 $, and $ {{\bf{w}}_\vartheta}[l] \in {\mathbb{C}^{M \times 1}} $ represents the sensing beamforming vector at the DFBS.

\subsection{Communication Model}
The received signal at the $ k $-th L-UAV in the $ l $-th time slot can be expressed as
\begin{align}
	{y_{k}}[l] = ({\bf{h}}_{d,k}^{\mathsf{H}}[l] + {\bf{h}}_{\mathrm{r},k}^{\mathsf{H}}[l]{\bf{\Phi}}[l]{\bf{G}}[l]){\bf{x}}[l] + {n_k}[l],
\end{align}
where $ {\bf{h}}_{d,k}[l] \in {\mathbb{C}^{M \times 1}} $ denotes the channel vector between the DFBS and the $k$-th L-UAV, $ {\bf{h}}_{\mathrm{r},k}[l] \in {\mathbb{C}^{N \times 1}} $ represents the channel vector between the RIS and the $k$-th L-UAV, and $ {\bf G}[l] \in {\mathbb{C}^{N \times M}} $ denotes the channel matrix from the DFBS to the RIS. $ {\bf{\Phi }}[l] \buildrel \Delta \over = \mathrm{diag}\{ {\phi _1[l]},{\phi _2[l]},...,{\phi _N[l]}\} \in {\mathbb{C}^{N \times N}} $ with $ {\phi _n[l]} = {e^{\mathrm{j}{\theta _n[l]}}} $ represents the phase shift matrix of the RIS, where $ {\theta _n[l]} \in \left[ {0,2\pi } \right) $. $ {n_k}[l] \sim {\cal C}{\cal N}(0,\sigma _k^2) $ denotes the additive white Gaussian noise (AWGN). 

We assume that the channels $ {\bf{G}}[l] $, $ {{\bf{h}}_{d,k}}[l] $, and $ {{\bf{h}}_{\mathrm{r},k}}[l] $ are modeled as Rician fading. Then, we can obtain
\begin{align}
    {\bf{G}}[l] = {\beta _{\mathrm{G}}}\left(\sqrt {\frac{\kappa_{\mathrm{G}} }{{\kappa_{\mathrm{G}} + 1}}}{\bm{\alpha }} _{\mathrm{G}}^{\mathrm{LoS}} + \sqrt {\frac{1}{{\kappa_{\mathrm{G}} + 1}}} {\bm{\alpha }} _{\mathrm{G}}^{\mathrm{NLoS}}[l]\right),
\end{align}
{\small \begin{align}
    {{\bf{h}}_{\mathrm{d},k}}[l] = {\beta _{\mathrm{d},k}}[l]\left(\sqrt {\frac{{{\kappa _{\mathrm{d},k}}}}{{{\kappa _{\mathrm{d},k}} + 1}}} {\bm{\alpha }}_{\mathrm{d},k}^{\mathrm{LoS}}[l] + \sqrt {\frac{1}{{{\kappa _{\mathrm{d},k}} + 1}}} {\bm{\alpha }}_{\mathrm{d},k}^{\mathrm{NLoS}}[l]\right),
\end{align}}%
and
{\small \begin{align}
    {{\bf{h}}_{\mathrm{r},k}}[l] = {\beta _{\mathrm{r},k}}[l]\left(\sqrt {\frac{{{\kappa _{\mathrm{r},k}}}}{{{\kappa _{\mathrm{r},k}} + 1}}} {\bm{\alpha }}_{\mathrm{r},k}^{\mathrm{LoS}}[l] + \sqrt {\frac{1}{{{\kappa _{\mathrm{r},k}} + 1}}} {\bm{\alpha }}_{\mathrm{r},k}^{\mathrm{NLoS}}[l]\right),
\end{align}}%
where 
$\beta _{\mathrm{G}} = \frac{\beta_0}{\Vert {\bf{q}}_{\mathrm{BS}} - {\bf{q}}_{\mathrm{RIS}}\Vert}$,
$\beta _{\mathrm{d},k} [l] = \frac{\beta_0}{\Vert {\bf{q}}_{\mathrm{BS}} - {\bf{q}}_{k}[l]\Vert}$
% $ {\beta _{\Gamma }}[l]$ $(\Gamma \in \{ G, \{ d,k\}, \{ r,k\}) $ 
$\beta _{\mathrm{r},k} [l] = \frac{\beta_0}{\Vert {\bf{q}}_{\mathrm{RIS}} - {\bf{q}}_{k}[l]\Vert}$
represent the distance dependent path loss, and $\beta_0$ denotes the path loss at the reference distance $1 \mathrm{m}$. $ {{\bm{\alpha }}^{\mathrm{LoS}}_{\mathrm{G}}} = {\bf{a}}_N({\theta _{\mathrm{R,B}}}){{\bf{a}}_M^{\mathsf{H}}}({\theta _{\mathrm{B, R}}}) $ is the LoS component with the steering vectors $ {{\bf{a}}_N}({\theta _{\mathrm{R,B}}}) = {[1,{e^{ - {\mathrm{j}}\pi \sin {\theta _{\mathrm{R,B}}}}},...,{e^{ - {\mathrm{j}}(N - 1)\pi \sin {\theta _{\mathrm{R,B}}}}}]^{\mathsf{H}}} $, and $ {{\bf{a}}_M}({\theta _{\mathrm{B, R}}}) = {[1,{e^{ - {\mathrm{j}}\pi \sin {\theta _{\mathrm{B, R}}}}},...,{e^{ - {\mathrm{j}}(M - 1)\pi \sin {\theta _{\mathrm{B, R}}}}}} ]^{\mathsf{H}} $, where $ {\theta _{\mathrm{R,B}}}$ and $ {\theta _{\mathrm{B, R}}} $ denote the direct-of-arrival (DoA) and direct-of-departure of channel $ {\bf{G}} $, respectively. 
Meanwhile, $ \alpha _{{\rm{d}},k}^{{\rm{LoS}}}[l] = {\bf{a}}_M^{\mathsf{}}({\theta _{{\rm{BS}},k}}[l]) $ and $ {{\bm{\alpha }}^{\mathrm{LoS}}_{\mathrm{r},k}[l]} = {\bf{a}}_N^{\mathsf{}}({\theta _{{\rm{RIS}},k}}[l]) $ are the steering vectors from the DFBS to the $ k $-th L-UAV, and from the RIS to the $ k $-th L-UAV, respectively, where $ {\theta _{\mathrm{BS}, k}[l]} $ and $ {\theta _{\mathrm{RIS}, k}[l]} $ are the DoAs at the $ k $-th L-UAV corresponding to the direct link by the DFBS, and the reflecting link by the RIS, respectively.
$ {{\bm{\alpha }}^{\mathrm{NLoS}}_{\mathrm{G}}[l]} $, $ {{\bm{\alpha }}^{\mathrm{NLoS}}_{\mathrm{r},k}[l]} $, and $ {{\bm{\alpha }}^{\mathrm{NLoS}}_{\mathrm{d},k}[l]} $ represent the NLoS components, having a zero mean and unit variance. $ \kappa_{\mathrm{G}} $, $\kappa_{\mathrm{d},k}$, and $\kappa _{\mathrm{r},k}$ are the Rician factor. 

The decoded signal-to-interference-plus-nois-ratio of the $k$-th L-UAV can be denoted as
\begin{align}
	{\gamma _k}[l] = \frac{{{{\left| {{\bf{H}}_k^{\mathsf{H}}[l]{{\bf{w}}_k[l]}} \right|}^2}}}{{\sum\limits_{i = 1,i \ne k}^K {{\left| {{\bf{H}}_k^{\mathsf{H}}[l]{{\bf{w}}_i[l]}} \right|^2}} + {{\left| {{\bf{H}}_k^{\mathsf{H}}[l]{{\bf{w}}_\vartheta[l] }} \right|}^2} + \sigma_k^2}},
\end{align}
where $ {\bf{H}}_k^{\mathsf{H}}[l] \buildrel \Delta \over = {\bf{h}}_{\mathrm{d},k}^{\mathsf{H}}[l] + {\bf{h}}_{\mathrm{r},k}^{\mathsf{H}}[l]{\bf{\Phi}}[l] {\bf{G}}[l] $. Then, the data rate of the $ k $-th L-UAV can be expressed as
\begin{align}
	{R_k}[l] = {\log _2}(1 + {\gamma _k}[l]).
\end{align}

\subsection{Sensing Model}
The DFBS utilizes the same signals for performing communication services and sensing detections. Thus, the echo signal coming from the U-UAV in time slot $ l $ can be written as 
\begin{align}
	{{\bf {y}}_{t}}[l] = {{\bf{G}}_t}[l]{\bf{x}}[l] + {{\bf{n}}_{t}}[l],
\end{align}
where $ {{\bf{G}}_t}[l] \buildrel \Delta \over = ({{\bf{g}}_{\mathrm{d},t}}[l] + {{\bf{G}}^{\mathsf{H}}[l]}{\bf{\Phi }}[l]{{\bf{g}}_{\mathrm{r},t}}[l])({\bf{g}}_{\mathrm{d},t}^{\mathsf{H}}[l] + {\bf{g}}_{\mathrm{r},t}^{\mathsf{H}}[l]{\bf{\Phi}}[l]\\ {\bf{G}}[l]) $, $ {\bf{g}}_{\mathrm{d},t}[l] \in {\mathbb{C}^{M \times 1}} $ and $ {\bf{g}}_{\mathrm{r},t}[l] \in {\mathbb{C}^{N \times 1}} $ represent the channel vectors from the DFBS to the U-UAV and from the RIS to the U-UAV, respectively. Adopting a common assumption in radar sensing \cite{10364735,luoRISAided2023}, we model the propagation paths between the DFBS/RIS and the U-UAV as LoS channels. Specifically, these channels can be expressed as $ {{\bf{g}}_{\mathrm{d},t}}[l] = {\beta _{\mathrm{d},t}[l]}{{\bf{a}}_M^{\mathsf{}}}(\theta _t^1[l]) $ and $ {{\bf{g}}_{\mathrm{r},t}}[l] = {\beta _{\mathrm{r},t}[l]}{{\bf{a}}_N^{\mathsf{}}}(\theta _t^2[l]) $,
where $\theta _t^1[l] $ and $ \theta _t^2[l] $ denote the DoAs of the U-UAV with respect to the DFBS and the RIS, respectively. $ {\beta _{\mathrm{d},t}}[l] = \frac{{{\beta _0}}}{{\left\| {{{\bf{q}}_{{\rm{BS}}}} - {{\bf{q}}_{\rm{U}}}[l]} \right\|}} $ and $ {\beta _{\mathrm{r},t}}[l] = \frac{{{\beta _0}}}{{\left\| {{{\bf{q}}_{{\rm{RIS}}}} - {{\bf{q}}_{\rm{U}}}[l]} \right\|}} $ denote the corresponding distance dependent path loss.
$ {{\bf{n}}_{t}}[l] \sim {\cal C}{\cal N}({\bf{0}},\sigma_{t}^2{{\bf{I}}_M}) $ denotes the AWGN. Then, the radar sensing SNR for the U-UAV is given by
\begin{align}
	{\rm{SN}}{{\rm{R}}_t}[l] = \frac{{\sum\limits_{i = 1}^K {{{\left| {{\bf{G}}_t^{\mathsf{H}}[l]{{\bf{w}}_i[l]}} \right|}^2}}  + {{\left| {{\bf{G}}_t^{\mathsf{H}}[l]{{\bf{w}}_\vartheta[l] }} \right|}^2}}}{{\sigma _t^2}}.
\end{align}\(\)

% For the successful detection of the U-UAV, only when the sensing SNR of the U-UAV is higher than a specific threshold will its effective decoding ability of legitimate signals pose a substantial security threat.
%The presence of the U-UAV poses a significant security threat only when its sensing SNR exceeds a specific threshold.
%Hence, this paper optimizes the resource allocation of the system under the premise of satisfying the SNR constraint of the U-UAV, effectively improving the system communication performance in the presence of unauthorized intruders and providing a solution for the security design for the LANs.

The presence of the U-UAV poses a significant security threat only when its sensing SNR exceeds a specific threshold. Therefore, this paper focuses on optimizing the system's resource allocation while ensuring that the U-UAV's SNR constraint is satisfied. This approach effectively enhances the system's communication performance in the presence of unauthorized intruders and offers a viable solution for the security design of LANs.

\subsection{Problem Formulation}
% A joint optimization framework for RIS-assisted ISAC systems is proposed in this paper, 
We aim to optimize the sum-rate of all L-UAVs, by co-designing the communication beamforming ${\bf{w}}_k[l]$, the sensing beamforming ${\bf{w}}_\vartheta[l]$, and the RIS phase shift parameters ${\bf{\Phi }}[l]$.
Then, the optimization problem can be formulated as
\begin{subequations}
	\begin{align}
		{\rm{P1}}:& \mathop {\max }\limits_{{\bf{w}}_k[l], {\bf{w}}_\vartheta[l], {\bf{\Phi }}[l]}\ \frac{1}{L}\sum\limits_{l = 1}^L {\sum\limits_{k = 1}^K {{R_k}[l]} }, \label{27a}\\
		\mathrm{s.t.}&\ \ \frac{{\sum\limits_{i = 1}^K {{{\left| {{\bf{G}}_t^{\mathsf{H}}[l]{{\bf{w}}_i[l]}} \right|}^2}}  + {{\left| {{\bf{G}}_t^{\mathsf{H}}[l]{{\bf{w}}_\vartheta[l] }} \right|}^2}}}{{\sigma _t^2}} \ge {\Gamma}, \forall{l}, \label{27b} \\
		&\left| {{\phi _n}[l]} \right| = 1, n \in {\cal N}, \forall{l}, \label{27d}\\
		&{\sum\limits_{k = 1}^K {\left\| {\bf{w}}_k[l] \right\|^2} + \left\| {\bf{w}}_\vartheta[l] \right\|^2} \le P_{\max}, \forall{l}, \label{27e}
	\end{align}
\end{subequations}
where \(\Gamma\) denotes the pre-defined radar sensing SNR requirement and \(P_{\max}\) is the maximum transmit power of the DFBS.
Constraint \eqref{27b} represents the sensing requirement.
Constraint \eqref{27d} is the physical constriant of the element phase parameters of the RIS.
Lastly, constraint \eqref{27e} denotes the transmit power constraint of the DFBS.

Evidently, problem P1 exhibits non-convexity. The reason is that the variables are coupled in both the objective function and constraint (\ref{27b}).
Moreover, the constraint of the RIS phase parameters to unitary modes exacerbates the challenge of solving P1.

\section{Coordinated Beamforming For RIS-Empowered ISAC System}
% Consequently, i
In the subsequent discussion, we put forward the employment of FP and AO methods to convert problem P1 into three sub-problems, which we solve through an iterative process. 

\subsection{Transformation of Objective Function}
In the optimization problem P1, due to the coupling relationship between the optimization variables in the numerator and denominator of the objective function, the problem is difficult to solve directly. Therefore, according to \cite{8310563}, we utilize the FP method to handle this obstacle. Specifically, we first introduce the auxiliary variable ${\bf{r}}[l] \buildrel \Delta \over = {[{r_1[l]},{r_2[l]},...,{r_K[l]}]^{\mathsf{T}}}$, and then restructure the objective function as
\begin{align}
    f =&\ \frac{1}{L}\sum\limits_{l = 1}^L \left( \sum\limits_{k = 1}^K {{{\log }_2}(1 + {r_k}[l])} - \sum\limits_{k = 1}^K {{r_k}[l]} \notag \right.\\
    &+ \left. \sum\limits_{k = 1}^K {\frac{{(1 + {r_k}[l]){{\left| {{\bf{H}}_k^{\mathsf{H}}[l]{{\bf{w}}_k}[l]} \right|}^2}}}{{\sum\limits_{i = 1}^K {{{\left| {{\bf{H}}_k^{\mathsf{H}}[l]{{\bf{w}}_i}[l]} \right|}^2}} + {{\left| {{\bf{H}}_k^{\mathsf{H}}[l]{{\bf{w}}_\vartheta }[l]} \right|}^2} + \sigma _k^2}}} \right). \label{21}
\end{align}

\begin{Proposition}
The original optimization problem is mathematically equivalent to the following reformulated problem
\begin{subequations}
	\begin{align}
		{\rm{P2}}: \mathop {\max }\limits_{{\bf{w}}_k[l], {\bf{w}}_\vartheta[l], {\bf{\Phi }}[l], {\bf{r}}[l]}\ &\eqref{21} \nonumber \\
		\mathrm{s.t.}\qquad\quad &\eqref{27b}\text{-}\eqref{27e} \nonumber.
	\end{align}
\end{subequations}
\begin{proof} 
    Notice that when \( {\bf w}[l] \) is fixed, the function in \eqref{21} is a differentiable concave function with respect to \( {\bf r}[l] \). As such, the optimal solution for \( {\bf r}[l] \) can be derived by setting each \( \frac{\partial f}{\partial {\bf r}[l]} \) to zero. Once the optimized \( {\bf r}[l] \) is substituted back into \eqref{21}, the objective function in \eqref{27a} is precisely recovered. This process confirms the equivalence between the two problems {\rm{P1}} and {\rm{P2}}.
    %The detailed proof of this equivalence can be found in \cite{8310563}.
\end{proof}
\end{Proposition}

% \textbf{\textit{Proposition 1:}} The original problem is equivalent to
% \begin{subequations}
% 	\begin{align}
% 		{\rm{P2}}: \mathop {\max }\limits_{{\bf{w}}_k[l], {\bf{w}}_\vartheta[l], {\bf{\Phi }}[l], {\bf{r}}[l]}\ &(\ref{21}), \\
% 		s.t.\qquad\quad &\eqref{27b}-\eqref{27e},
% 	\end{align}
% \end{subequations}

% \textbf{\textit{Proof:}} The detailed constructive proof can be found in \cite{8310563}.

To efficiently solve the problem P2, we decouple it into three disjoint subproblems and target to find the solution of each subproblem. Detailed information are given as follows.

\subsection{Block Update}

\subsubsection{Closed-Form Optimal Solution for $ {\bf r}[l] $} When other variables are fixed, the optimal value of $ {\bf{r}}[l] $ can be derived by solving $ \frac{{\partial f}}{{\partial {\bf{r}}[l]}} = 0 $, leading to the expression
\begin{align}
	r_k^*[l] = \frac{{{{\left| {{\bf{H}}_k^{\mathsf{H}}[l]{{\bf{w}}_k[l]}} \right|}^2}}}{{\sum\limits_{i = 1,i \ne k}^K {{{\left| {{\bf{H}}_k^{\mathsf{H}}[l]{{\bf{w}}_i[l]}} \right|}^2}}  + {{\left| {{\bf{H}}_k^{\mathsf{H}}[l]{{\bf{w}}_\vartheta[l] }} \right|}^2} + \sigma _k^2}}. \label{n9}
\end{align}

\subsubsection{Update Active Beamforming} Given fixed $ {\bf r}[l] $ and $ {\bf{\Phi}}[l] $, the optimization problem for active beamforming can be derived by isolating relevant terms, resulting in the following problem P3

\begin{subequations}
	{\small \begin{align} 
            {\rm{P3}}:&\mathop {\max }\limits_{{{\bf{w}}_k[l]},{{\bf{w}}_\vartheta[l] }} \;\frac{1}{L}\sum\limits_{l = 1}^L {\sum\limits_{k = 1}^K {\frac{{(1 + {r_k[l]}){{\left| {{\bf{H}}_k^{\mathsf{H}}[l]{{\bf{w}}_k}[l]} \right|}^2}}}{{\sum\limits_{i = 1}^K {{{\left| {{\bf{H}}_k^{\mathsf{H}}[l]{{\bf{w}}_i}[l]} \right|}^2}}  + {{\left| {{\bf{H}}_k^{\mathsf{H}}[l]{{\bf{w}}_\vartheta }[l]} \right|}^2} + \sigma _k^2}}} }, \\
            \mathrm{s.t.}& \quad \frac{{\sum\limits_{i = 1}^K {{{\left| {{\bf{G}}_t^{\mathsf{H}}[l]{{\bf{w}}_i[l]}} \right|}^2}}  + {{\left| {{\bf{G}}_t^{\mathsf{H}}[l]{{\bf{w}}_\vartheta[l] }} \right|}^2}}}{{\sigma _t^2}} \ge {\Gamma}, \forall{l},\label{12b}\\
            &\quad \sum\limits_{k = 1}^K {{{\left\| {{{\bf{w}}_k[l]}} \right\|}^2}} + {\left\| {{{\bf{w}}_\vartheta[l] }} \right\|^2} \le {P_{\max}},\forall{l}.
	\end{align}}%
\end{subequations}
Since the objective function of problem P3 exhibits a sum-of-ratio structure, we utilize the quadratic transformation method and introduce the auxiliary variables $ {\bf{c}}[l] \buildrel \Delta \over = {[{c_1}[l],{c_2}[l],...,{c_K}[l]]^{\mathsf{T}}} $. Then, we can reformulate the objective function as
\begin{align}
    &{\cal F}({\bf{w}},{\bf{c}}) \buildrel \Delta \over = \frac{1}{L}\sum\limits_{l = 1}^L \sum\limits_{k = 1}^K \Bigg(2\sqrt {1 + {r_k}[l]} \Re \{ c_k^*[l]{\bf{H}}_k^{\mathsf{H}}[l]{{\bf{w}}_k}[l]\} \notag \\
    & -\left. {{\left| {{c_k}[l]} \right|}^2}\left(\sum\limits_{i = 1}^K {{{\left| {{\bf{H}}_k^{\mathsf{H}}[l]{{\bf{w}}_i}[l]} \right|}^2}}  + {{\left| {{\bf{H}}_k^{\mathsf{H}}[l]{{\bf{w}}_\vartheta }[l]} \right|}^2} + \sigma _k^2\right)\right),
\end{align}
where 
\begin{align}
    c_k^*[l] = \frac{{\sqrt {1 + {r_k[l]}} {\bf{H}}_k^{\mathsf{H}}[l]{{\bf{w}}_k[l]}}}{{\sum\limits_{i = 1}^K {{{\left| {{\bf{H}}_k^{\mathsf{H}}[l]{{\bf{w}}_i[l]}} \right|}^2}}  + {{\left| {{\bf{H}}_k^{\mathsf{H}}[l]{{\bf{w}}_\vartheta[l] }} \right|}^2} + \sigma _k^2}}. \label{n12}
\end{align}
To handle the quadratic terms in (\ref{12b}), we introduce the following Proposition 2.
\begin{Proposition}
For the quadratic term defined by
\begin{align}
    I({\bf{w}}[l]) = {\left| {{\bf{G}}_t^{\mathsf{H}}[l]{{\bf{w}}}[l]} \right|^2},\label{nn15}
\end{align}
its first-order Taylor expansion to the point $ {\widehat {\bf{w}}}[l] $ is given by
\begin{align}
    \widehat I({\bf{w}}[l],\widehat {\bf{w}}[l]) =&\ 2\Re \{\widehat {\bf{w}}^{\mathsf{H}}[l]{{\bf{G}}_t[l]}{\bf{G}}_t^{\mathsf{H}}[l]{{\bf{w}}[l]}\} \notag \\
    &- \Re \{\widehat {\bf{w}}^{\mathsf{H}}[l]{{\bf{G}}_t[l]}{\bf{G}}_t^{\mathsf{H}}[l]{\widehat {\bf{w}}[l]}\}.
\end{align}
Therefore, $ I({\bf{w}}[l]) $ can be approximated by $ \widehat I({\bf{w}}[l],\widehat {\bf{w}}[l]) $ at point $ {\widehat {\bf{w}}}[l] $.
\end{Proposition}
\begin{proof}
Using the Taylor series approximation, the quadratic term defined in \eqref{nn15} satisfies
\begin{align}
    I({\bf{w}}[l]) \ge I(\widehat {\bf{w}}[l]) + {\left. {\frac{{\partial I}}{{\partial {\bf{w}}[l]}}} \right|_{\widehat {\bf{w}}[l]}}({\bf{w}}[l] - \widehat {\bf{w}}[l]).
\end{align}
Then, we have
\begin{align}
    I({\bf{w}}[l]) \ge&\ {\widehat {\bf{w}}^{\mathsf{H}}}[l]{{\bf{G}}_t}[l]{\bf{G}}_t^{\mathsf{H}}[l]\widehat {\bf{w}}[l] \notag \\
    &+ 2{\widehat {\bf{w}}^{\mathsf{H}}}[l]{{\bf{G}}_t}[l]{\bf{G}}_t^{\mathsf{H}}[l]({\bf{w}}[l] - \widehat {\bf{w}}[l])\notag \\
    \ge&\ 2{\widehat {\bf{w}}^{\mathsf{H}}}[l]{{\bf{G}}_t}[l]{\bf{G}}_t^{\mathsf{H}}[l]{\bf{w}}[l] \notag \\
    &- {\widehat {\bf{w}}^{\mathsf{H}}}[l]{{\bf{G}}_t}[l]{\bf{G}}_t^{\mathsf{H}}[l]\widehat {\bf{w}}[l].
\end{align}
Considering the real characteristic of $ I({\bf{w}}[l]) $, \eqref{nn15} can be rewritten as
\begin{align}
    I({\bf{w}}[l]) \ge&\ 2\Re \{ {\widehat {\bf{w}}^{\mathsf{H}}}[l]{{\bf{G}}_t}[l]{\bf{G}}_t^{\mathsf{H}}[l]{\bf{w}}[l]\} \notag \\
    &- \Re \{ {\widehat {\bf{w}}^{\mathsf{H}}}[l]{{\bf{G}}_t}[l]{\bf{G}}_t^{\mathsf{H}}[l]\widehat {\bf{w}}[l]\} \notag \\
    \buildrel \Delta \over =&\ \widehat I({\bf{w}}[l],\widehat {\bf{w}}[l]).
\end{align}
The proof is completed.
\end{proof}
%we concentrate on the quadratic term $ {{{\left| {{\bf{G}}_t^H[l]{{\bf{w}}_i[l]}} \right|}^2}} $. The first-order Taylor approximation of this term around $\widehat {\bf{w}}_i[l]$ can be formulated as 
%\begin{align}
    %{\left| {{\bf{G}}_t^H[l]{{\bf{w}}_i[l]}} \right|^2} &= 2\Re \{\widehat {\bf{w}}_i^H[l]{{\bf{G}}_t[l]}{\bf{G}}_t^H[l]{{\bf{w}}_i[l]}\} \notag \\
    %&- \Re \{\widehat {\bf{w}}_i^H[l]{{\bf{G}}_t[l]}{\bf{G}}_t^H[l]{\widehat {\bf{w}}_i[l]}\},
%\end{align}
\noindent Then, (\ref{12b}) can be reformulated as
{\small \begin{align}
	&\sum\limits_{i = 1}^K {2\Re \{\widehat {\bf{w}}_i^{\mathsf{H}}[l]{{\bf{G}}_t[l]}{\bf{G}}_t^{\mathsf{H}}[l]{{\bf{w}}_i[l]}\} - \Re \{\widehat {\bf{w}}_i^{\mathsf{H}}[l]{{\bf{G}}_t[l]}{\bf{G}}_t^{\mathsf{H}}[l]{{\widehat {\bf{w}}}_i[l]}\}} \notag \\
	&+ 2\Re \{\widehat {\bf{w}}_\vartheta ^{\mathsf{H}}[l]{{\bf{G}}_t[l]}{\bf{G}}_t^{\mathsf{H}}[l]{{\bf{w}}_\vartheta[l] }\} - \Re \{\widehat {\bf{w}}_\vartheta ^{\mathsf{H}}[l]{{\bf{G}}_t[l]}{\bf{G}}_t^{\mathsf{H}}[l]{\widehat {\bf{w}}_\vartheta[l] }\} \notag \\
    & \ge \Gamma \sigma _t^2.
\end{align}}%
Based on the derivatives and transformations above, the optimization problem P3 has been transformed into a convex one.

\subsubsection{Passive Beamforming Optimization}
Define $ {{\bf{v}}^{\mathsf{H}}[l]} = [{v_1}[l],{v_2}[l],...,{v_N}[l]] $ where $ {v_n}[l] = {e^{\mathrm{j}{\theta _n}[l]}} $, we can reconstruct the complex channel as
\begin{align}
    \begin{array}{l}
    ({\bf{h}}_{\mathrm{d},i}^{\mathsf{H}}[l] + {\bf{h}}_{\mathrm{r},i}^{\mathsf{H}}[l]{\bf{\Phi }}[l]{\bf{G}}[l]){{\bf{w}}_j}[l] = \\
    \underbrace {\left[ {{\bf{h}}_{\mathrm{r},i}^{\mathsf{H}}[l]{\rm{diag}}({\bf{G}}[l]{{\bf{w}}_j}[l])\ \ {\bf{h}}_{\mathrm{d},i}^{\mathsf{H}}[l]{{\bf{w}}_j}[l]} \right]}_{\widetilde {\bf{H}}_{i,j}^{\mathsf{H}}[l]}\underbrace {\left[ {\begin{array}{*{20}{l}}
    {{\bf{v}}[l]}\\
    {\;1}
    \end{array}} \right]}_{\widetilde {\bf{v}}[l]},i,j \in \{ K,\vartheta \},
    \end{array}
\end{align}
and
\begin{align}
    &{\bf{G}}_t^{\mathsf{H}}[l]{{\bf{w}}_i}[l] \notag \\ 
    &=\left( {\underbrace {\left[ {{{\bf{G}}^{\mathsf{H}}}[l]{\rm{diag}}({{\bf{g}}_{\mathrm{r},t}}[l])\ \ {{\bf{g}}_{\mathrm{d},t}}[l]} \right]}_{{{\widetilde {\bf{G}}}_t}[l]}\widetilde {\bf{v}}[l]} \right)\left( {{{\widetilde {\bf{v}}}^{\mathsf{H}}}[l]\widetilde {\bf{G}}_t^{\mathsf{H}}[l]} \right){{\bf{w}}_i}[l] \notag \\
    &= {\widetilde {\bf{G}}_t}[l]\widetilde {\bf{v}}[l]{\widetilde {\bf{v}}^{\mathsf{H}}}[l]\widetilde {\bf{G}}_t^{\mathsf{H}}[l]{{\bf{w}}_i}[l], i \in \{ K,\vartheta \},
\end{align}
Then, the passive beamforming optimization problem can be written as
\begin{subequations}
	\begin{align}
		&{\rm{P4}}:\mathop {\max }\limits_{\widetilde{\bf{v}}[l]} \;         {\cal F}(\widetilde{\bf{v}}[l]),\\
		&\mathrm{s.t.}\ \sum\limits_{i = 1}^K {{{\left| {{\bf{w}}_i^{\mathsf{H}}[l]           {{\widetilde {\bf{G}}}_t[l]}\widetilde {\bf{v}}[l]{{\widetilde       {\bf{v}}}^{\mathsf{H}}[l]}\widetilde {\bf{G}}_t^{\mathsf{H}}[l]} \right|}^2}}           \notag \\
            &\qquad + {\left| {{\bf{w}}_\vartheta ^{\mathsf{H}}[l]{{\widetilde {\bf{G}}}_t[l]}\widetilde {\bf{v}}[l]{{\widetilde {\bf{v}}}^{\mathsf{H}}[l]}\widetilde {\bf{G}}_t^{\mathsf{H}}[l]} \right|^2} \ge \Gamma \sigma _t^2, \label{52e}\\
		&\qquad \left| {{v_n}[l]} \right| = 1,n \in {\cal N}, \label{52h}
	\end{align}
\end{subequations}
where 
\begin{align}
    &{\cal F}(\widetilde {\bf{v}}[l]) = \frac{1}{L}\sum\limits_{l = 1}^L \sum\limits_{k = 1}^K \Bigg(2\sqrt {1 + {r_k}[l]} \Re \{ c_k^*[l]\widetilde {\bf{H}}_{k,k}^{\mathsf{H}}[l]\widetilde {\bf{v}}[l]\}  \notag \\
    &- \left. {{\left| {{c_k}[l]} \right|}^2}\left(\sum\limits_{i = 1}^K {{{\left| {\widetilde {\bf{H}}_{k,i}^{\mathsf{H}}[l]{\rm{ }}\widetilde {\bf{v}}[l]} \right|}^2}}  + {{\left| {\widetilde {\bf{H}}_{k,\vartheta }^{\mathsf{H}}[l]{\rm{ }}\widetilde {\bf{v}}[l]} \right|}^2} + \sigma _k^2\right)\right).
\end{align}
To address the quadratic terms in the left-hand side of (\ref{52e}), we introduce the following Proposition 3.
\begin{Proposition}
For the quadratic term 
\begin{align}
    L = {\left| {{\bf{w}}^{\mathsf{H}}[l]{{\widetilde {\bf{G}}}_t[l]}\widetilde {\bf{v}}[l]{{\widetilde {\bf{v}}}^{\mathsf{H}}[l]}\widetilde {\bf{G}}_t^{\mathsf{H}}[l]} \right|^2},\label{nn25}
\end{align}
its first-order Taylor expansion on point $ \overline {\widetilde {\bf{v}}}[l] $ yields
\begin{align}
    \bar L({\bf w}[l]) =&\ 4\Re \{ {\widetilde {\bf{G}}_t}[l]{\widetilde {\bf{v}}}[l]{\overline {\widetilde {\bf{v}}} ^{\mathsf{H}}}[l]{\bf{E}}[l]\overline {\widetilde {\bf{v}}} [l]{\overline {\widetilde {\bf{v}}} ^{\mathsf{H}}}[l]\widetilde {\bf{G}}_t^{\mathsf{H}}[l]\}  \notag \\
    &- 3\Re \{ {\widetilde {\bf{G}}_t}[l]\overline {\widetilde {\bf{v}}} [l]{\overline {\widetilde {\bf{v}}} ^{\mathsf{H}}}[l]{\bf{E}}[l]\overline {\widetilde {\bf{v}}} [l]{\overline {\widetilde {\bf{v}}} ^{\mathsf{H}}}[l]{\bf{G}}_t^{\mathsf{H}}[l]\},
\end{align}
where $ {\bf{E}}[l] = \widetilde {\bf{G}}_t^{\mathsf{H}}[l]{\bf{w}}[l]{{\bf{w}}^{\mathsf{H}}}[l]{\widetilde {\bf{G}}_t}[l] $. Thus, $ L({\bf w}[l]) $ can be approximated by $ \bar L({\bf w}[l]) $ at $ \overline {\widetilde {\bf{v}}}[l] $.
\end{Proposition}
\begin{proof}
The proof process is similar to that of Proposition 2. Thus the detailed proof is omitted here.
\end{proof}

%define $ L({\bf{w}}_i[l]) = {\left| {{\bf{w}}_i^H[l]{{\widetilde {\bf{G}}}_t[l]}\widetilde {\bf{v}}[l]{{\widetilde {\bf{v}}}^H[l]}\widetilde {\bf{G}}_t^H[l]} \right|^2} $, the first-order Taylor expansion on $ \overline {\widetilde {\bf{v}}}[l] $ can be written as \cite{10681754}
%\begin{align}
    %\overline L ({\bf{w}}_i[l]) &= 4\Re\{{\widetilde {\bf{G}}_t[l]}\widetilde {\bf{v}}[l]{\bf{E}}{\overline {\widetilde {\bf{v}}} ^H[l]}\widetilde {\bf{G}}_t^H[l]\} \notag \\
    %&- 3\Re\{{\widetilde {\bf{G}}_t[l]}\overline {\widetilde {\bf{v}}}[l] {\bf{E}}{\overline {\widetilde {\bf{v}}} ^H[l]}\widetilde {\bf{G}}_t^H[l]\},
%\end{align}
%where $ {\bf{E}} \buildrel \Delta \over = {\overline {\widetilde {\bf{v}}} ^H[l]}\widetilde {\bf{G}}_t^H[l]{{\bf{w}}_i[l]}{\bf{w}}_i^H[l]{\widetilde {\bf{G}}_t[l]}\overline {\widetilde {\bf{v}}}[l] $. Thus, $ L({\bf{w}}_i[l]) $ can be approximated by $ \overline L ({\bf{w}}_i[l]) $ at the extreme $ \overline {\widetilde {\bf{v}}}[l] $. 
\noindent Then, (\ref{52e}) can be converted to
\begin{align}
    \sum\limits_{i = 1}^K {\overline L ({{\bf{w}}_i}[l])}  + \overline L ({{\bf{w}}_\vartheta[l]}) \ge \Gamma \sigma _t^2.
\end{align}
We can observe that the key challenge in optimizing the subproblem P4 lies in the unit-modulus constraint (\ref{52h}). This constraint can be effectively addressed using the penalty convex-concave procedure proposed in \cite{10472878}. In particular, we can reformulate the unit-modulus constraint as
\begin{align}
    1 \le {\left| {{v_n}[l]} \right|^2} \le 1,n \in {\cal N},
\end{align}
Based on Proposition 2, the part $ 1 \le {\left| {{v_n}[l]} \right|^2} $ can be recast as
\begin{align}
    1 \le 2\Re \{\widehat v_n^{\mathsf{H}}[l]{v_n[l]}\}- \Re \{\widehat v_n^{\mathsf{H}}[l]{\widehat v_n[l]}\}, n \in {\cal N}.
\end{align}
Following the aforementioned manipulations, problem P4 becomes a convex optimization problem.
% A comprehensive overview of the methodology for tackling problem P1 is presented in Algorithm 1. 

\subsection{Convergence Analysis}
Building upon the discussions presented above, this paper introduces an iterative algorithm to address problem P2, which is concisely summarized in Algorithm 1. The convergence property of the proposed Algorithm 1 is established in the following theorem.

\begin{Proposition}
\label{theorem 1}
Algorithm 1 guarantees that the objective function value of problem P2 does not decrease during each iteration and will eventually reach a converged point.
\end{Proposition}

\begin{proof}
The detailed proof is provided in Appendix A.
\end{proof}

\subsection{Computational Complexity}
Finally, we briefly analyze the computational complexity of the proposed joint beamforming design algorithm. As shown in Algorithm 1, the computational burden mainly results from the update for $ {\bf w} $ and $ {\bf \Phi} $, whose complexities are of order $ {\cal O}({{((K + 1)M)^3}}) $ and $ {\cal O}({{N^{3.5}}}) $, respectively. Thus, the overall complexity of the proposed algorithm is of order $ {\cal O}({{((K + 1)M)^3}} + N^{3.5}) $. 

\begin{algorithm}[t]
	\caption{The Proposed Algorithm to Solve Problem P1.}
		\begin{algorithmic}[1]
			\STATE
			\textbf{Initialize:} $ {\bf \Phi}^{(0)} $, $ {\bf w}_k^{(0)} $, and $ {\bf{w}}_\vartheta^{(0)} $, iteration index $ p=1 $ and accuracy threshold $ \varepsilon > 0 $.
			\STATE
			\textbf{Repeat}
			\STATE
			Update $ {\bf r}^{(p)} $ according to (\ref{n9});
			\STATE
			Update $ {\bf c}^{(p)} $ accroding to (\ref{n12});
			\STATE
			Update $ {\bf w}_k^{(p)} $ and $ {\bf{w}}_\vartheta^{(p)} $ by solving problem P3;
			\STATE
			Update $ {\bf \Phi}^{(p)} $ by solving problem P4;
                \STATE
                \textbf{Until} the increase of the objective function between two
 adjacent iterations in problem P1 is smaller than $ \varepsilon $.
			\label{algorithm1}
		\end{algorithmic}
\end{algorithm}

\section{Numerical Results}
%\begin{figure}[t]
	%\centering
	%\includegraphics[width=0.92\linewidth,height=0.66\linewidth]{"BP"}
	%\caption{Enhanced beampattern of the RIS-assisted ISAC system. (DFBS: diamond; RIS: square; T-UAV: star; L-UAVs: circles).}
	%\label{fig:3}
%\end{figure}
This section presents numerical results to demonstrate the effectiveness of the proposed joint active and passive beamforming design. The system under consideration comprises a DFBS with $ M=6 $ antennas, four single-antenna L-UAVs, and one single-antenna U-UAV. The noise variances at the U-UAV and L-UAVs are set to $ \sigma _t^2 = \sigma _k^2 = - 90\mathrm{dBm}, \forall{k} $. A distance-dependent path-loss model \cite{8811733} is adopted with path-loss exponents of $ 2.2 $ (DFBS-RIS and RIS-U-UAV links), $ 2.3 $ (RIS-L-UAV links), $ 2.4 $ (DFBS-U-UAV link), and $ 3.5 $ (DFBS-L-UAV links). Because of the severe channel fading caused by the larger distance between the L-UAVs and the U-UAV, the reflected signals from the U-UAV to the L-UAVs are disregarded. All DFBS-RIS and RIS-L-UAV channels follow Rician fading with factors $ \kappa_{\mathrm{G}} $, $\kappa_{\mathrm{d},k}$, and $\kappa _{\mathrm{r},k} = 3\mathrm{dB} $. The maximum transmit power and radar SNR threshold are constrained to $ 33 \mathrm{dBm} $ and $ 4 \mathrm{dB} $, respectively.

To demonstrate the effectiveness of the proposed algorithm, two comparison schemes are considered. 
\begin{itemize}	
	\item Random phase: In this scheme, the phase shifts of the RIS are random given.
	
	\item Without RIS: In this scheme, the communication and sensing links supported by the RIS are excluded.
\end{itemize}
For all schemes, we set different antenna numbers (i.e., $ M=6 $ and $ M=8 $) for comparison.
\begin{figure}[t]
	\centering
	\includegraphics[width=0.92\linewidth,height=0.66\linewidth]{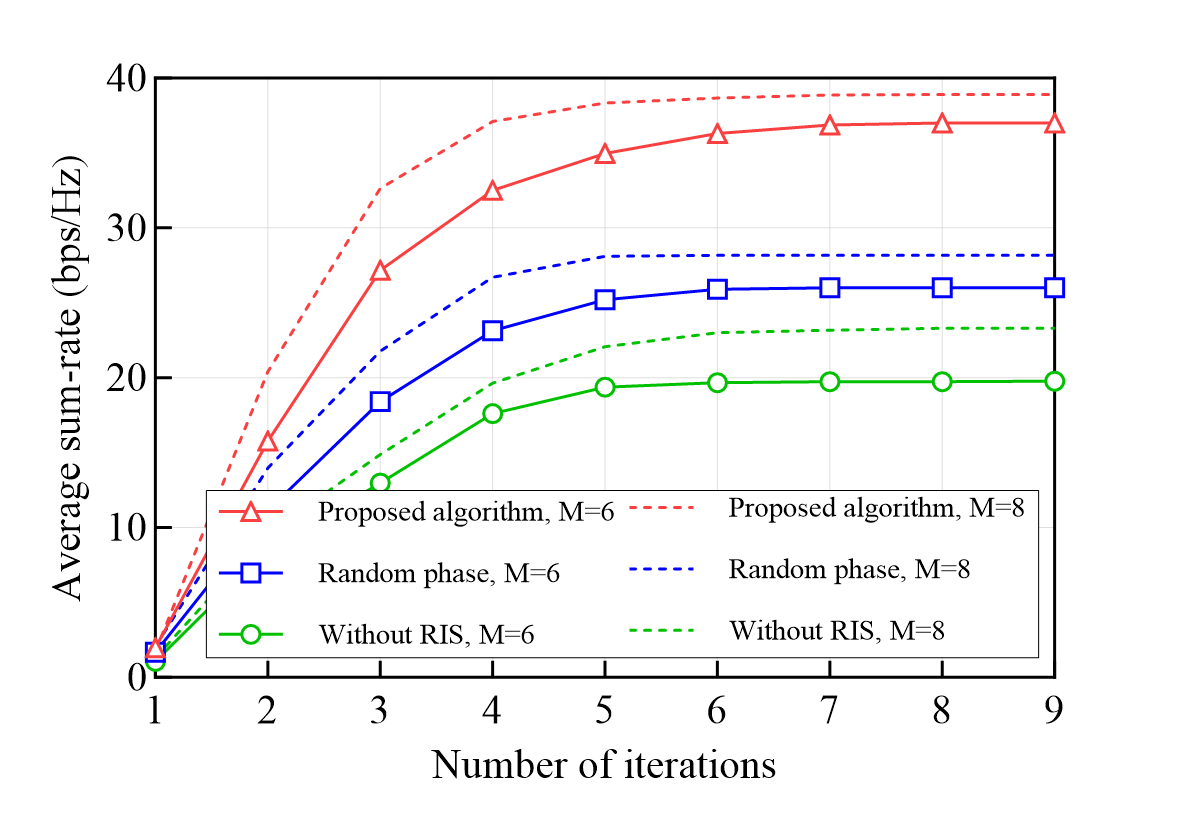}
	\caption{The convergence behaviors of the proposed algorithm and the baseline algorithms.}
	\label{fig:4}
\end{figure}
\begin{figure}[t]
	\centering
	\includegraphics[width=0.92\linewidth,height=0.66\linewidth]{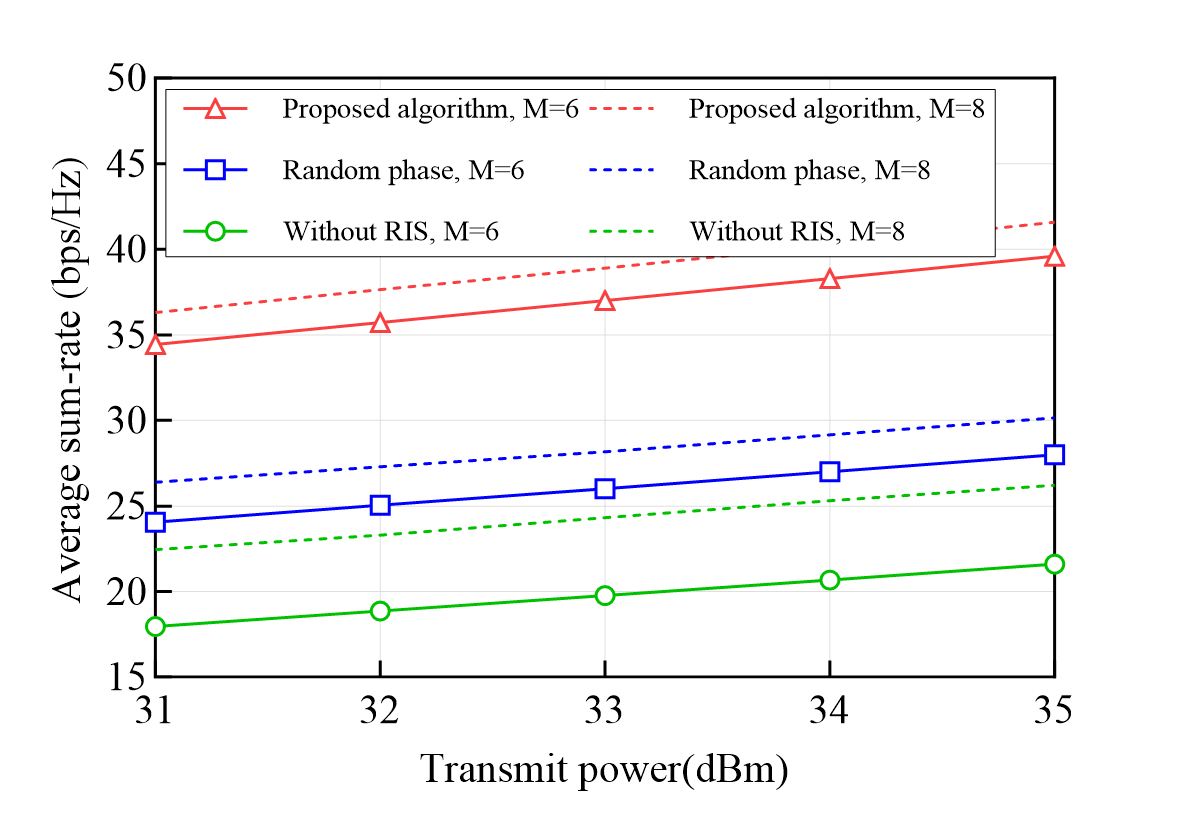}
	\caption{Average sum-rate versus the transmit power at the DFBS.}
	\label{fig:5}
\end{figure}
\begin{figure}[t]
	\centering
	\includegraphics[width=0.92\linewidth,height=0.66\linewidth]{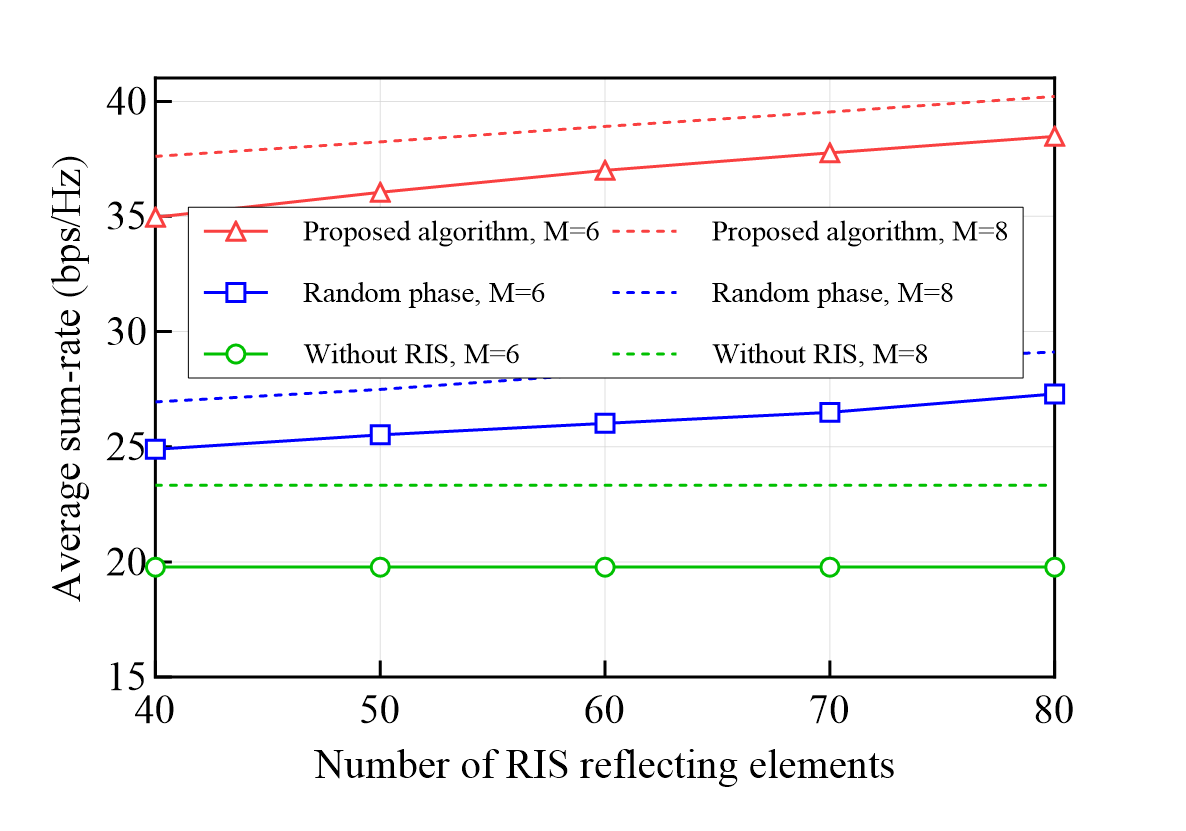}
	\caption{Average sum-rate versus the number of RIS reflecting elements.}
	\label{fig:6}
\end{figure}
\begin{figure}[t]
	\centering
	\includegraphics[width=0.92\linewidth,height=0.66\linewidth]{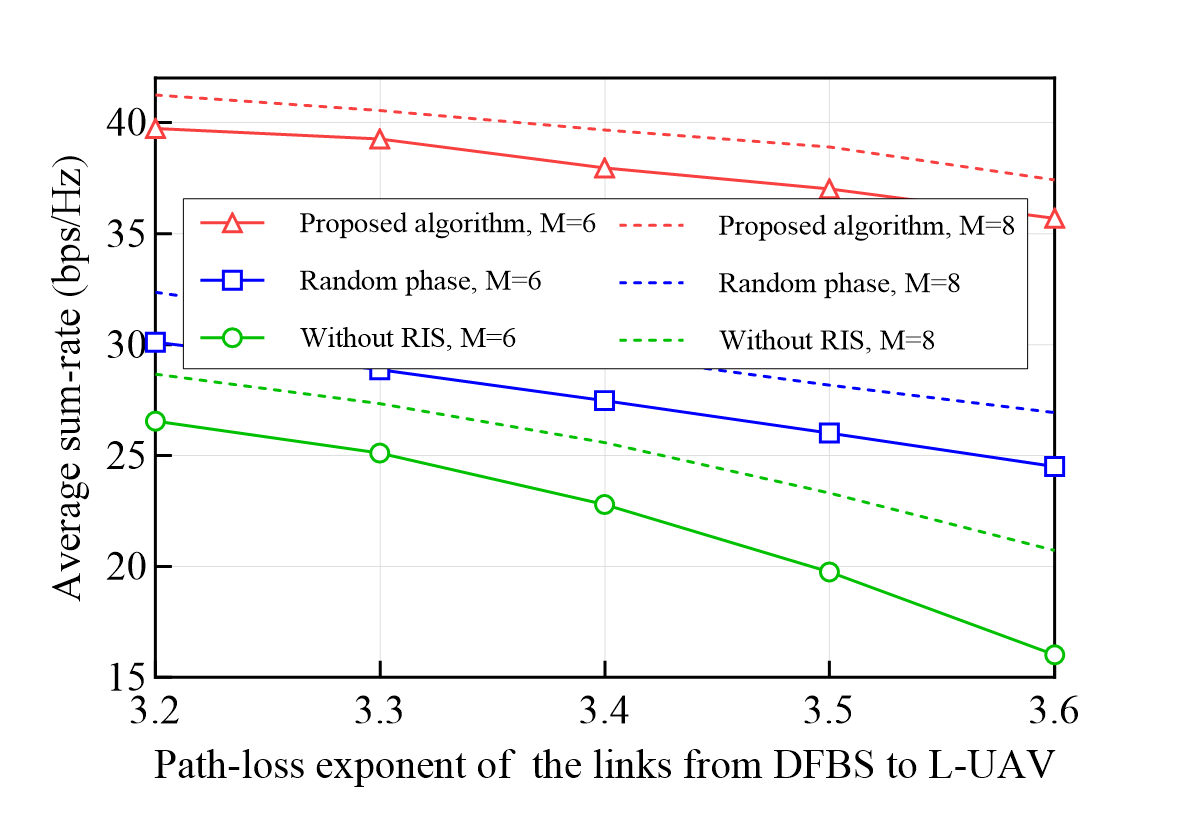}
	\caption{Average sum-rate versus the path-loss exponent of the DFBS-L-UAV links.}
	\label{fig:7}
\end{figure}
As illustrated in Fig. \ref{fig:4}, the convergence analysis reveals notable performance differences across various algorithms. 
Firstly, we can observe that all algorithms achieve convergence within 9 iterations, demonstrating their efficiency in terms of computational complexity. When comparing the proposed algorithm to the baseline algorithm (without RIS), a substantial improvement in average sum-rate is observed. This enhancement is primarily attributed to the introduction of additional NLoS links through the RIS, which facilitates passive beamforming gains. 
Furthermore, the proposed algorithm outperforms the random phase algorithm. This superior performance gain is due to the proposed algorithm's ability to optimize the RIS phase, whereas the random phase algorithm does not exploit this optimization potential.
Additionally, algorithms with $M=8$ antennas consistently outperform their $M=6$ counterparts. This result highlights the significant role of increasing the number of antennas $M$, as it leads to enhanced beamforming gain and increased spatial resources, thereby improving the overall transmission rate.

%Fig. \ref{fig:5} presents the sum-rate of all algorithms in relation to the transmit power $ P_{\max} $. We can see that as \(P_{\max}\) increases, the sum-rate also increases monotonically. The main reason for this is that a higher $ P_{\max} $ enhances the received signal power at the L-UAVs, thereby improving the SNR. It is further observed that under the same conditions, the RIS-aided algorithms show more remarkable enhancements. This is because the RIS has the ability to optimize signal propagation through its configurable links.

Fig. \ref{fig:5} depicts the relationship between the transmit power threshold $P_{\max}$ and the average sum-rate performance across all evaluated algorithms. As observed, the sum-rate increases monotonically with rising $P_{\max}$, indicating that higher transmit power leads to improved communication throughput. This performance gain is primarily attributed to the increase in received signal power at the L-UAVs, which in turn enhances the SNR. Moreover, the proposed algorithm consistently outperforms the baseline algorithm (without RIS) across all power levels. This performance gain stems from the RIS’s capability to intelligently reconfigure the wireless propagation environment by introducing additional reflective paths and enabling passive beamforming. As a result, the RIS not only amplifies the beneficial signal components but also mitigates interference, thereby maximizing the overall system efficiency. 

%In Fig. \ref{fig:6}, we plot the variation of the sum-rate with respect to the number of RIS elements. It can be found that for the proposed algorithm and random phase algorithm (whether with $ M = 6 $ or $ M = 8 $), the sum-rate exhibits a growing trend as the number of RIS reflecting elements increases. This is because more reflecting elements can provide greater channel gains. These additional gains optimize the signal propagation paths, which in turn drives the sum-rate to increase.
Fig. \ref{fig:6} illustrates the variation of the average sum-rate as a function of the number of RIS reflecting elements. As observed, both the proposed algorithm and the random phase algorithm (with $ M = 6 $ or $ M = 8 $) demonstrate an increasing average sum-rate with the growth in the number of RIS elements. This trend can be attributed to the enhanced channel gains provided by the additional reflecting elements. Specifically, more RIS elements enable the optimization of virtual signal propagation paths, thereby improving both the direct and reflected signal components. As a result, the overall communication link quality is improved, leading to a corresponding increase in the sum-rate. 

\begin{figure}[t]
	\centering
	\includegraphics[width=0.92\linewidth,height=0.66\linewidth]{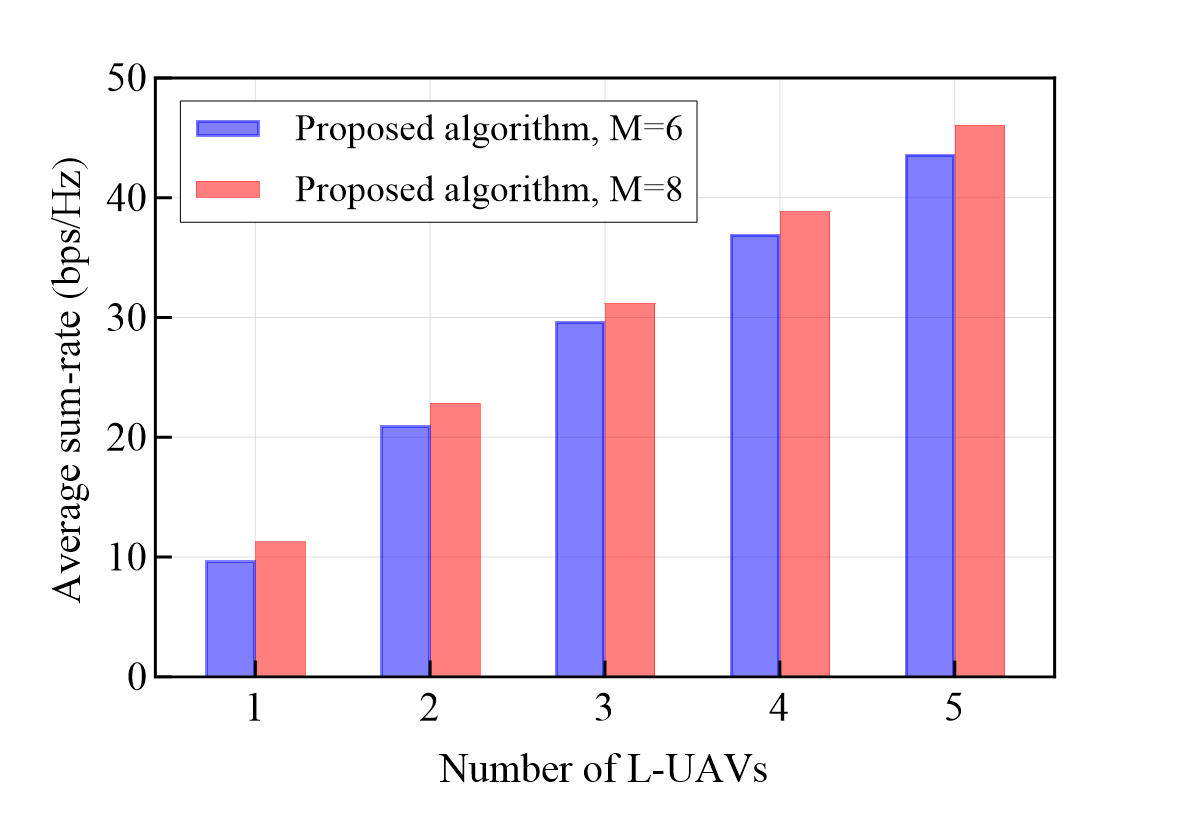}
	\caption{Average sum-rate versus the number of L-UAVs.}
	\label{fig:8}
\end{figure}

%Fig. \ref{fig:7} depicts the impact of the DFBS to L-UAV path-loss exponents on the optimized sum-rate. It can be observed that the sum-rate gradually decreases with the increase of the path-loss exponent. In addition, the sum-rate of the algorithm without RIS drops more severely under a larger path-loss exponent. This indicates that when the direct link is severely blocked, the performance improvement of the links constructed by RIS is more important.
Fig. \ref{fig:7} illustrates the effect of the DFBS-to-L-UAV path-loss exponent on the optimized average sum-rate. As shown, the average sum-rate decreases progressively with an increase in the path-loss exponent. Furthermore, the performance of the algorithm without RIS experiences a more significant decline as the path-loss exponent increases. This suggests that, in scenarios where the direct link suffers from substantial blockage, the performance improvement enabled by the RIS-constructed links becomes increasingly critical. Specifically, RIS plays a pivotal role in compensating for the attenuation of the direct link by providing alternative signal propagation paths, thereby mitigating the adverse effects of severe path-loss and sustaining system performance. 

%Fig. \ref{fig:8} illustrates that the system's sum-rate increases with the growth of the number of L-UAVs. This growth is mainly due to the fact that as the number of L-UAVs increases, the system's ability to reuse spectrum resources is enhanced. It can allocate data transmission tasks to more L-UAVs in the same frequency band, significantly improving the spectrum utilization efficiency and achieving more efficient data transmission.

Fig. \ref{fig:8} demonstrates that the average sum-rate improves as the number of L-UAVs increases. This growth can be primarily attributed to the enhanced spectrum resource reuse capability with a higher number of L-UAVs. As the number of L-UAVs expands, the system can allocate the data transmission tasks to a greater number of UAVs within the same frequency band. Consequently, more efficient data transmission is achieved, thereby boosting the overall system performance. 

\section{Conclusion}
In this paper, we investigated the joint optimization of active beamforming and passive beamforming in a RIS-empowered ISAC system for secure LANs. The proposed framework aimed to maximize the sum-rate of L-UAVs while guaranteeing the radar SNR constraint for detecting the U-UAV, the transmit power limitations, and the unit modulus properties of the RIS phase shift parameters. By leveraging the FP and the AO methods, the non-convex problem was decomposed into three tractable subproblems, enabling efficient updates of the auxiliary variables, the active beamforming at the DFBS, and the passive beamforming at the RIS. Numerical simulations demonstrated that the joint optimization scheme significantly outperformed disjoint optimization benchmarks in terms of communication sum-rate, highlighting the critical role of the RIS in harmonizing communication and sensing functionalities for secure LANs. 

\begin{appendices}
\section{Proof of Proposition \ref{theorem 1}}
\begin{proof}
We prove Proposition \ref{theorem 1} by using a method similar to the ones employed in \cite{8811733,10288203}. Firstly, for simplify, we use $ {{\bf{w}}_k}, {{\bf{w}}_\vartheta }, {\bf{\Phi }} $, and $ {\bf{r}} $ to represent the optimization variables in time slot $ l $. Then, we define $ {R_{{P_2}}}({{{\bf{w}}_k^{(z)}},{{\bf{w}}_\vartheta^{(z)} },{\bf{\Phi }}^{(z)},{\bf{r}}^{(z)}}) $, 
$ {R_{{P_3}}}({{{\bf{w}}_k^{(z)}},{{\bf{w}}_\vartheta^{(z)} },{\bf{\Phi }}^{(z)},{\bf{r}}^{(z)}}) $, and $ {R_{P_4}}({{{\bf{w}}_k^{(z)}},{{\bf{w}}_\vartheta^{(z)} },{\bf{\Phi }}^{(z)},{\bf{r}}^{(z)}}) $ as the respective objective function values for problem P2, P3, and P4 in the $ z $-th iteration. Then, in the $ (z+1) $-th iteration, concerning the problem of transmit power and time optimization, the following holds
\begin{align}
    &R_{{P_2}}({{\bf{r}}^{(z)}, {{\bf{w}}_k^{(z)}}, {{\bf{w}}_\vartheta^{(z)} }, {\bf{\Phi }}^{(z)}}) \notag \\
    &\mathop = \limits^{(a)} {R_{{P_3}}}({{\bf{r}}^{(z)}, {{\bf{w}}_k^{(z)}}, {{\bf{w}}_\vartheta^{(z)} }, {\bf{\Phi }}^{(z)}}) \notag \\
    &\mathop \le \limits^{(b)} {R_{{P_3}}}({{\bf{r}}^{(z+1)}, {{\bf{w}}_k^{(z+1)}}, {{\bf{w}}_\vartheta^{(z+1)} }, {\bf{\Phi }}^{(z)}}) \notag \\
    &= {R_{{P_2}}}({{\bf{r}}^{(z+1)}, {{\bf{w}}_k^{(z)}},{{\bf{w}}_\vartheta^{(z)} },{\bf{\Phi }}^{(z)}}),\label{37}
\end{align}
where $ (a) $ holds due to the problems P2 and P3 have the same solution; and $ (b) $ holds as $ ({{\bf{r}}^{(z+1)}, {{\bf{w}}_k^{(z+1)}}, {{\bf{w}}_\vartheta^{(z+1)} }}) $ denotes the optimal solution to problem P3 in the $ (z+1) $-th iteration. 
	
For the IRS beamforming optimization problem in the $ (z+1) $-th iteration, following the same analysis in \eqref{37}, we have
\begin{align}
    &R_{{P_2}}({{\bf{r}}^{(z+1)}, {{\bf{w}}_k^{(z+1)}}, {{\bf{w}}_\vartheta^{(z+1)} }, {\bf{\Phi }}^{(z)}}) \notag \\
    &= {R_{{P_4}}}({{\bf{r}}^{(z+1)}, {{\bf{w}}_k^{(z+1)}}, {{\bf{w}}_\vartheta^{(z+1)} }, {\bf{\Phi }}^{(z)}}) \notag \\
    &\le {R_{{P_4}}}({{\bf{r}}^{(z+1)}, {{\bf{w}}_k^{(z+1)}}, {{\bf{w}}_\vartheta^{(z+1)} }, {\bf{\Phi }}^{(z+1)}}) \notag \\
    &= {R_{{P_2}}}({{\bf{r}}^{(z+1)}, {{\bf{w}}_k^{(z+1)}},{{\bf{w}}_\vartheta^{(z+1)} },{\bf{\Phi }}^{(z+1)}}),\label{38}
\end{align}
	
\noindent Based on \eqref{37} and \eqref{38}, we can conclude
\begin{align}
    &R_{{P_2}}({{\bf{r}}^{(z)}, {{\bf{w}}_k^{(z)}}, {{\bf{w}}_\vartheta^{(z)} }, {\bf{\Phi }}^{(z)}}) \notag \\
    &\le {R_{{P_2}}}({{\bf{r}}^{(z+1)}, {{\bf{w}}_k^{(z+1)}},{{\bf{w}}_\vartheta^{(z+1)} },{\bf{\Phi }}^{(z+1)}}).
\end{align}
Hence, the objective function value of problem P2 maintains its non-decreasing trend throughout the iterative process. Since the variables are bounded, the objective function value of problem P2 finally converges.
\end{proof}
	
\end{appendices}

\bibliographystyle{IEEEtran}
\bibliography{reference}

\end{document}